\documentclass[leqno]{amsart}
\usepackage{amsmath}
\usepackage{amssymb}
\usepackage{amsthm}
\usepackage{enumerate}
\usepackage[mathscr]{eucal}
\usepackage{graphicx}
\usepackage{color}

\usepackage{centernot}
\usepackage{mathtools}
\usepackage{ stmaryrd }

\theoremstyle{plain}
\newtheorem{theorem}{Theorem}[section]

\newtheorem{lemma}{Lemma}[section]

\theoremstyle{definition}

\newtheorem{example}{Example}

\usepackage[pagewise]{lineno}

\usepackage{lipsum}
\begin{document}	
\title[Binary LCD group codes]{Classification and count of binary linear complementary dual group codes}
\author[Shaw, Bhowmick, Bagchi]{Ankan Shaw, Sanjit Bhowmick, Satya Bagchi}
\newcommand{\acr}{\newline\indent}

\address[Shaw]{Department of Mathematics\\ National Institute of Technology Durgapur\\ Durgapur 713209\\ West Bengal\\ INDIA}
\email{ankanf22@gmail.com}

\address[Bhowmick]{Department of Mathematics\\ National Institute of Technology Durgapur\\ Durgapur 713209\\ West Bengal\\ INDIA}
\email{sanjitbhowmick392@gmail.com}

\address[Bagchi]{Department of Mathematics\\ National Institute of Technology Durgapur\\ Durgapur 713209\\ West Bengal\\ INDIA}
\email{satya.bagchi@maths.nitdgp.ac.in}

\thanks{The research of Mr. Ankan Shaw is supported by CSIR HRDG in terms of Junior Research Fellowship. This work was also supported by National Board of Higher Mathematics (NBHM), Government
of India (Grant No. 02011/2/2019 NBHM(R.P)/R\&D II/1092).} 


\begin{abstract}
We establish a complete classification of binary group codes with complementary duals for a finite group and explicitly determine the number of linear complementary dual (LCD) cyclic group codes by using cyclotomic cosets. The dimension and the minimum distance for LCD group codes are explored. Finally, we find a connection between LCD MDS group codes and maximal ideals.  
\end{abstract}

\maketitle


\section{Introduction}
A linear code $ C $ of length $ n $ is $ \mathbb{F}_2$-linear  subspace of $ \mathbb{F}_2^n $. An element $ c = (c_1, c_2, \dots, c_n)\in C $ is called a codeword of $ C $. The (Hamming) weight of $ c $ is denoted by $ wt(c) $ and defined by $ wt(c) = |\{ i\in\{1, 2, \dots, n\} : c_i\ne 0 \}|$. The minimum distance of a linear code $ C $ is denoted by $ d(C) $ and defined by $ d(C) = \min \left\{ wt(c) : c\in C \setminus \{0\}  \right\}$. An $[ n, k, d ] $ code is a linear code of length $ n $, dimension $ k $ and minimum distance $ d $ over $ \mathbb{F}_2 $.
 
 Through out this article, we consider $ \mathcal{G} $ is a finite group of order $ n $ with the identity element $ g_0 $. A group algebra $ \mathbb{F}_2 \mathcal{G}  $ denotes the collection of all formal sum such that  $  \mathbb{F}_2  \mathcal{G}  = \left\{ \sum_{ g\in  \mathcal{G} } a_g g : a_g\in \mathbb{F}_2 \right\}$ endowed with the binary operations $ `` + " $ and $ `` * " $. For $ a = \sum_{ g\in  \mathcal{G} } a_g g $ and $ b = \sum_{ g\in  \mathcal{G} } b_g g $ belong to $ \mathbb{F}_2  \mathcal{G}  $, the binary operations $ `` + " $ and $ `` * " $ are defined as $ a + b =  \sum_{ g\in  \mathcal{G} } (a_g + b_g ) g $ and $ a * b = ab = \sum_{ g\in  \mathcal{G} } \,(\sum_{ h\in G} a_h b_{h^{-1} }g ) \,g $ respectively.
 
We note that, under these operations, $ \mathbb{F}_2  \mathcal{G}  $ becomes a $ \mathbb{F}_2$-vector space with the element of $ \mathcal{G} $ as a basis element. We can easily see that, $ ( \mathbb{F}_2  \mathcal{G}  , + , * )$ is a ring containing $0$ and $g_0$ in its center. Therefore, $ \mathbb{F}_2  \mathcal{G}  $ becomes a $ \mathbb{F}_2$-algebra  via multiplication.

After ordering the elements of $ \mathcal{G} $ as $ g_0, g_1, \dots, g_{n-1} $, we can define a map
\begin{equation}\label{paper1_equn1}
 \mathbb{F}_2  \mathcal{G} \ni g_i \xmapsto{\;\;\Phi\;\;} e_{i+1} \in \mathbb{F}_2 ^n,
\end{equation}

where $ \{e_1, e_2, \dots, e_n \} $ is the standard basis of $ \mathbb{F}_2 ^n $. Here $\Phi$ is an isomorphism. Therefore, the group algebra $  \mathbb{F}_2  \mathcal{G}  $ is isomorphic to $ \mathbb{F}_2^n $ as a $ \mathbb{F}_2$-vector space. In this way, we can transfer many coding theoretical properties from $ \mathbb{F}_2 ^n $ to $ \mathbb{F}_2  \mathcal{G}  $. We extend the above map $ \mathbb{F}_2$-linear  so that $\Phi :\sum_{i=0}^{n-1} a_i g_i \longmapsto ( a_0, a_1, \dots, a_{n-1} )$. Note that the isomorphism $ \Phi $ is not canonical because it depends on the ordering of the group $ \mathcal{G} $. Since different ordering of the elements of $ \mathcal{G} $ leads only to a permutation of the co-ordinates, hence codes are permutation equivalent.

From the above isomorphism, we can transfer the Hamming metric from $ \mathbb{F}_2 ^n $ to $ \mathbb{F}_2  \mathcal{G}  $. Thus for $ a, b \in \mathbb{F}_2  \mathcal{G} $, we define $ wt(a) = wt ( \Phi(a) )$ and a standard inner product of $a$ and $b$ is as follows: \[ \langle a, b\rangle : = \langle \Phi(a),\Phi(b) \rangle = \sum a_g\cdot b_g \in \mathbb{F}_2.\] So this extends the classical duality to the group algebra context. Therefore, from a coding theoretical point of view, we consider linear codes either in $ \mathbb{F}_2  \mathcal{G}  $ or in $ \mathbb{F}_2 ^n $ without any difference.

A group code $ \mathcal{C} $ is a right ideal in $ \mathbb{F}_2  \mathcal{G}   $. If $  \mathcal{G}  $ is abelian (cyclic), then $ \mathcal{C} $ is called abelian (cyclic) group code. The adjoint of $ a \in \mathbb{F}_2  \mathcal{G}  $, is defined as $ \hat a = \sum_{ g\in  \mathcal{G} } a_g g^{-1} $. An element $ a $ is called self-adjoint if $ a = \hat a $. A binary linear code $ [ n , k , d ] $ is said to be maximum distance separable (MDS) if $ n - k +1 = d $. For an element $ a = \sum_{i=0}^{n-1} a_i g_i $ in $  \mathbb{F}_2  \mathcal{G}  $,  $ a_i  $'s are the coefficients of $ a $ and $ g_i $'s are the components of $ a $ for $ i \in [n] := \{1, 2, \dots, n\}$. In this article, we consider $ \mathcal{C} = e \mathbb{F}_2 \mathcal{G}  $ as an LCD  group code and component set of $ e $ as $ M $. The $2$-cyclotomic  coset modulo $n$ containing $i$ is defined by \[ C_i = \left\{i\cdot 2^j \pmod{n} \colon j = 0, 1, 2, \dots \right\} .\] Thus we can partition the group $\mathbb{Z}_{n}$ into cyclotomic cosets $ C_i $, where $i\in \mathbb{N}\cup\{0\}$.

A linear code $C$ is called linear complementary dual (LCD) if $ C\cap C^\perp = \{ 0 \} $. In 1967, Berman \cite{Berman1967} proposed the concept of group codes and then studied abelian group codes and Reed-Muller codes using finite group representations. After a long time, in \cite{delacruz2018}, Cruz and Willems gave a relation between group codes and LCD codes. In \cite{Massey1992}, it was shown that LCD codes provide an optimum linear coding solution for binary adder channel. In \cite{Yang}, Yang and Massey have given a necessary and sufficient condition for a cyclic code to have a complementary dual. The LCD condition for a certain class of quasi-cyclic codes has been studied in \cite{Esmaeilis2009}. In \cite{Carlet2016}, Dougherty gave a linear programming bound on the largest size of an LCD code of given length and minimum distance. In \cite{Li2018}, Li constructed some non-MDS cyclic Hermitian LCD codes over a finite field and analyzed their parameters. In \cite{Pang2018}, a class of MDS negacyclic LCD codes of even length was given. Carlet and Guilley studied an application of LCD codes against side-channel attacks and presented particular constructions for LCD codes in \cite{Carlet2016}. In \cite{Beelen2018}, Beelen and Jin gave an explicit construction of several classes of LCD MDS codes over a finite field $ \mathbb{F}_q $ which were completely solved for even $ q $ in \cite{Jin2017}. In \cite{LiLi2017}, authors explored two special families of cyclic codes, which are both BCH codes. In \cite{Chen2018}, Chen and Liu proposed a different approach to obtain new LCD MDS codes from generalized Reed-Solomon codes. In \cite{Sok2017}, authors proved the existence of optimal LCD codes over a large finite field, and they also gave methods to generate orthogonal matrices over a finite field and then use them to construct LCD codes. In \cite{Carlet2018}, Carlet studied several constructions of new Euclidean and Hermitian LCD MDS codes. There are few articles on group codes over finite fields. With a special criterion, the authors \cite{Bernal2009} classified the Cauchy codes which are left group codes. In \cite{Borello2020} and \cite{delacruz2018}, the authors studied LCD group codes and provided a necessary and sufficient condition for an LCD group code and concluded with an open problem.     

The rest of the paper is organized as follows. In Section \ref{Sec2}, a characterization of binary LCD group codes is given. Binary LCD properties of group codes for abelian groups are discussed in Section \ref{Sec3}. In Section \ref{Sec4}, we count the number of LCD cyclic group codes. We have shown the dimension and the minimum distance of a group code in Section \ref{Sec5}. A condition is investigated for MDS group codes in Section \ref{Sec6}. Some suitable examples are given in Section \ref{Sec7}. We conclude the paper in Section \ref{Sec8}.

\section{Characterization of LCD group codes}\label{Sec2}
Throughout this section, we consider $\mathcal{G} $ is any finite group. The main contribution of this section is that $M$ does not contain $2$ order element. 

\begin{lemma}\label{Paper1_lemma3.1}
A group code $\mathcal{C}$ is LCD if and only if $\Phi (\mathcal{C}) $ is LCD over $\mathbb{F}_2$.
\end{lemma}

\begin{proof}
From the definition of dot product of the group code, we get $\Phi(\mathcal{C})^\perp  = \Phi(\mathcal{C}^{\perp})$. Then $x\in \Phi(\mathcal{C}) \cap \Phi ( \mathcal{C})^\perp $ if and only if $ x \in \Phi(\mathcal{C} \cap \mathcal{C}^\perp)$. This completes the proof.
\end{proof}

We get the following theorem due to Massey \cite{Massey1992}.
\begin{theorem}\label{Paper1_Theorem3.1}
Let $G$ be a generator matrix of a code $C$ over $\mathbb{F}_2$. Then $ C $ is LCD if and only if $ G G^T $ is invertible.
\end{theorem}

\begin{theorem}
Let $\mathcal{C}$ be a group code with a free basis $ B = \{b_1, b_2, \dots, b_k\}$. Then $\mathcal{C}$ is an LCD group code if and only if $P$ is invertible, where the matrix $$P =  \begin{bmatrix}
\langle b_1, b_1 \rangle & \langle b_1, b_2\rangle   & \dots  & \langle b_1, b_k\rangle \\
\langle b_2, b_1\rangle  & \langle b_2, b_2\rangle   & \dots  & \langle b_2, b_k\rangle \\
 \vdots                  &    \vdots                 & \vdots & \vdots  \\
\langle b_k, b_1\rangle  & \langle b_k, b_2\rangle   & \dots  &  \langle b_k, b_k\rangle 
\end{bmatrix}.$$ 
\end{theorem}

\begin{proof}
We consider the isomorphism $\Phi : \mathbb{F}_2\mathcal{G} \rightarrow \mathbb{F}_2 ^n $. Then $\Phi (\mathcal{C}) $ is a linear code with a basis $\Phi (B) = \left\{\Phi (b_1), \Phi (b_2), \dots, \Phi (b_k) \right\} $. Then a generator matrix of the code $\Phi(\mathcal{C})$ is \[ G = \begin{bmatrix} \Phi (b_1) \\ \Phi (b_2) \\ \vdots \\ \Phi (b_k) \end{bmatrix}. \] 
Then \[ G G^T =  \begin{bmatrix}
\langle \Phi (b_1),\Phi (b_1) \rangle  & \langle \Phi (b_1), \Phi (b_2) \rangle  & \cdots   & \langle \Phi (b_1),\Phi (b_k) \rangle \\
\langle \Phi (b_2),\Phi (b_1)\rangle   & \langle \Phi (b_2),\Phi (b_2)\rangle    & \cdots   & \langle \Phi (b_2),\Phi (b_k) \rangle \\
 \vdots                  &            \vdots         & \vdots   &  \vdots     \\
\langle \Phi (b_k),\Phi (b_1)\rangle  & \langle \Phi (b_k), \Phi (b_2) \rangle    & \dots    &  \langle \Phi (b_k),\Phi (b_k) \rangle 
\end{bmatrix}. \]  Since  $\langle b_i, b_j\rangle =\langle \Phi (b_i), \Phi (b_j) \rangle $, so $P = G G^T$. From Lemma \ref{Paper1_lemma3.1}, $\Phi (\mathcal{C}) $ is LCD. By the Theorem \ref{Paper1_Theorem3.1}, $\Phi (\mathcal{C}) $ is LCD if and only if $P$ is invertible. Hence $\mathcal{C}$ is an LCD group code if and only if $P$ is invertible.
\end{proof}

We get the following theorem from De la Cruz and Willems.
\begin{theorem}\label{Paper1_Theorem3.4}\cite[Theorem 3.1]{delacruz2018}
If $ \mathcal{C} $ is a right ideal in $ \mathbb{F}_2 \mathcal{G} $, then the following are equivalent:
\begin{enumerate}
\item $\mathcal{C}$ is an LCD group code,
\item $ \mathcal{C} = e\mathbb{F}_2  \mathcal{G}  $, where $ e^2 = e = \hat{e} $.
\end{enumerate}
\end{theorem}
From the Theorem \ref{Paper1_Theorem3.4}, we see that $ \mathcal{C}^\perp = (1-e)\mathbb{F}_2  \mathcal{G}  $.

\begin{theorem}\label{Paper1_Theorem3.5}
Suppose $ \mathcal{C} = e \mathbb{F}_2  \mathcal{G}  $ is an LCD group code. Then the component set $M$ of $e$ does not contain $ 2 $ order element.
\end{theorem}

\begin{proof}
Let $  M = \{g_1, g_2, \dots, g_l \}  $ be the component set of $ e $. Since $ \mathcal{C} $ is an LCD  group code, so $ e^2 = e = \hat{e} $. Suppose that $ M $ contains  an element of order $ 2 $. So there exists a non-identity element $ g_r \in M $ such that $ g^2_r = g_0 $.

\textbf{Case I:} Suppose  $ g^2_1 = g_r$ which implies $ g^4_1 = g_0 $. This shows that the order of $ g_1 $ is $ 4 $. Since $  e = \hat{e} $, so  $ g_1 , g^3_1 \in M $. So, we have $ (g^3_1)^2 + (g^{-3}_1)^2 = 0 $. This shows that $ g^2_1\neq g_r$.

\textbf{Case II:} Again, without loss of generality, suppose $g_1 g_2 = g_r $. Since $  e = \hat{e} $, so $ g^{-1}_1,  g^{-1}_2 \in M $. Now $ (g_1 g_2)^{-1} = g^{-1}_r = g_r $. From above we get $ g_1 g_2 + (g_1 g_2)^{-1} = 0 $. This shows that $g_1 g_2 \neq g_r $. 

In both cases $g_r \notin M $. These arguments conclude that there does not have $2$ order element in  $M$.
\end{proof}
\begin{theorem}\label{Paper1_Theorem3.6}
Suppose $ M $ is a subgroup of $ \mathcal{G} $. If $2 \nmid |M|$, then $\mathcal{C}=e \mathbb{F}_2  \mathcal{G} $ forms an LCD group code in $\mathbb{F}_2\mathcal{G}$.
\end{theorem}

\begin{proof}
 Since $2\nmid |M|$, each element of $M$ has an  odd order. Suppose  $ e = g_1 + g_2 + \dots + g_l$, so $M=\{ g_1, g_2, \dots, g_l\} $. Since $M$ forms a subgroup so $M = \widehat M $, where $\widehat M = \{ g^{-1}_1,g^{-1}_2, \dots, g^{-1}_l\} $. Hence $ e = \widehat e $. Since $2 \nmid |M|$ and $M$ forms a subgroup, so $ e^2 = |M|(g_1+g_2+\dots+g_l) = e $. Therefore we get $e = e^2 = \widehat e $. Hence $\mathcal{C}= e \mathbb{F}_2  \mathcal{G}  $ forms an LCD group code. 
\end{proof}

\begin{theorem}
Suppose $ M $ is a subgroup of $ \mathcal{G} $. If  $\mathcal{C}=e \mathbb{F}_2  \mathcal{G}  $ forms an LCD group code in $\mathbb{F}_2  \mathcal{G}  $, then $ 2 \nmid |M|$. 
\end{theorem}

\begin{proof}
The proof is followed from Cauchy's theorem and Theorem \ref{Paper1_Theorem3.5}.
\end{proof}

\section{LCD group codes for abelian groups}\label{Sec3}
In this section, we consider $ \mathcal{G} $ is an abelian group of order $n$.

\begin{theorem}\label{Paper1_Theorem3.7}
Let $\mathcal{C}= e  \mathbb{F}_2  \mathcal{G} $ be an LCD group code. Then $ M $ has no element of order $2k$; $k \in \mathbb{N}$.
\end{theorem}

\begin{proof}
Since $ \mathcal{G} $ is commutative, so $M = M^2 = \widehat{M} $, where $M^2$ contains only square elements of $M$. Let us consider there exists an element $g\in M $ such that the order of $g$ is $2k$, for some $k\in \mathbb{N}$. Since  $M = M^2 = \widehat{M} $, so $ g ^{2l} = g$, for some $l\in \mathbb{N}$. This shows that $2k \mid 2l-1$, which is impossible. Hence $ M $ has no element of order $2k$; $k\in\mathbb{N}$.
\end{proof}

Let $C$ be a binary cyclic code of length $ 2 ^a $, $ a\in\mathbb{N}$ with the monic generator polynomial $ g(x)$. Then $ g(x)\mid (x^{2^a} -1) $. Since $(x^{2^a} -1) = (x - 1 )^{2^a} $ in $ \mathbb{F}_2[x] $, we get $ g(x) = (x-1) ^i $, $ 0 \leq i \leq 2^a $. Now consider a cyclic group $  \mathcal{G}  = \langle g \rangle $. 

Now, we define a mapping 
\begin{equation}\label{paper1_equn2}
\mathbb{F}_2  \mathcal{G} \ni g \xmapsto{\;\;\Phi_{1}\;\;} x \in \mathbb{F}_2 [x]/ \langle  x ^{2^a} - 1 \rangle,
\end{equation}
which is a module isomorphism.

\begin{theorem}
Every subgroup of a cyclic group $ \mathcal{G} $ of order $ 2^a $, there exists a bijective correspondence between group code $e \mathbb{F}_2  \mathcal{G}  $ and  a cyclic code of length $ 2^a $ over $ \mathbb{F}_2 $.
\end{theorem}

\begin{proof}
Let $ e = g_0 + g^{2^i} $. From Equation (\ref{paper1_equn2}), we get $\Phi_{1}(e) = 1+ x^{2^i} = g(x) $. Then $e \mathbb{F}_2  \mathcal{G}  $ forms a group code. Hence  $C = \langle  g(x) \rangle  $ is a linear code over $ \mathbb{F}_2 $ of length $ 2^a $ which  correspondences to $\mathcal{C} = e \mathbb{F}_2  \mathcal{G}  $ in $\mathbb{F}_2  \mathcal{G}  $.
\end{proof}

\begin{theorem}
For all  $a\in \mathbb{N}$, there does not exist nontrivial cyclic LCD code over $ \mathbb{F}_2 $ of length $ 2^a $.
\end{theorem}

\begin{proof}
Let $C$ be a nontrivial cyclic LCD code over $ \mathbb{F}_2 $ of length $ 2^a $. Then there exists a generator polynomial $g(x)$ such that $ C = \langle  g(x) \rangle $, $ g(x)\mid ( x^{2^a} -1 )$, and $g(x) = (x-1) ^i $ for some $i$; $1 \leq i \leq 2^a -1 $.

From Equation (\ref{paper1_equn2}), $\Phi^{-1}_{1}(g(x)) = g_0 + g^i$. Since $ C $ is nontrivial so $ g^i \neq g_0 $. Then from Lemma \ref{Paper1_lemma3.1}, $\mathcal{C} = e \mathbb{F}_2  \mathcal{G}  $ is an LCD  group code such that $ e = g_0 + g^{i}$. Here $\Phi_{1} (e) = g (x) $. Since $ \mathcal{C} $ is LCD, so $ e^2 = e $. Now $ e^2 = e $ implies that $ g^i = g_0 $ which contradicts that  $ g^i \neq g_0 $.  
This completes the proof.
\end{proof}

\section{Number of LCD cyclic group codes}\label{Sec4}
Throughout this section, let us suppose that $g$ is a generator of a cyclic group $ \mathcal{G} $ of odd order $n$. We know that $ \mathcal{G}\cong \mathbb{Z}_n $. There is no nontrivial self inverse element in $ \mathcal{G}$ as $(2,n) = 1$. So each non-identity element in $ \mathcal{G} $ forms a distinct pair with its inverse. Suppose $\mathcal{S}  = \{C_{r_1}, C_{r_2}, \dots, C_{r_t} \}$ is the complete set of distinct cyclotomic cosets of $2$ modulo $n$. So $|\mathcal{S}| = t $. We make partitions of $ \mathcal{G} $ into $t$ parts such that $ \mathcal{G}  = \bigcup\limits_{i=1}^{t}P_{r_i}$, where $P_{r_i}=\{g^j:j\in\ C_{r_i}\}$. Let $\mathcal{C} = e \mathbb{F}_2 \mathcal{G} $ be an LCD  group code and $M$ be the component set of $e$. 

\begin{lemma}\label{Paper1_Lemma6.1}
If $P_i\cap M \neq \phi$, then $P_i\subseteq M$.
\end{lemma}
\begin{proof}
Suppose $h\in\ P_i\cap M$. Then there exists $l \in \mathbb{N}\cup\{0\}$ such that $ h = g^{i2^{l} \pmod{n}} \in M$. Then by the definition of $P_{i}$, we get  $g^{i2^{l+j} \pmod{n}} \in P_{i}$ for all $j\in\ \mathbb{N}\cup\{0\} $. Since $ e^2 = e $, so $ g^{i2^{l+j} \pmod{n}} \in M $. Hence the lemma is proved. 
\end{proof}

\begin{theorem}
Suppose $ \mathcal{G} $ is a cyclic group of odd order $n$ and $n\mid (2^j+1)$ for some $j\in \mathbb{N}$. Then the number of LCD  group codes in $\mathbb{F}_2 \mathcal{G} $ is $(2^t-1)$.
\end{theorem}

\begin{proof}
From Lemma \ref{Paper1_Lemma6.1}, if $P_i\cap M \neq \phi$, $P_i\subseteq M$ for all $i\in [t]:=\{r_1, r_2, \dots, r_t\}$. We see that each $P_{i}$ is closed under square. Now, we claim that each $P_{i}$ is closed under inverse. Let $h \in P_{i}$. Then there exists $l \in \mathbb{N}\cup\{0\}$ such that $h = g^{i2^{l} \pmod{n}}$. Since $n \mid (2^j+1)$ for some $j \in \mathbb{N}$, so $2^j \equiv -1 \pmod{n}$. Then $g^{i2^{l + j} \pmod{ n}}$ is the inverse of $h$, which  belongs to $P_{i}$. So $P_{i}$ is closed under inverse. Therefore, we can take group elements of every nonempty  member of the power set $P(\mathcal{P})$ of $ \mathcal{P}=\{ P_{r_1},P_{r_2}, \dots, P_{r_t}\}$ as $M$. So we get total $(2^t-1)$ such $M$. Thus for each such $ M $, we can construct exactly one LCD  group code in $\mathbb{F}_2 \mathcal{G} $. Hence the total number of LCD group codes in $\mathbb{F}_2 \mathcal{G} $ is $(2^t-1)$.
\end{proof}

\begin{lemma}\label{Paper1_Lemma6.2}
$ 7\nmid (2^m+1)$ for all $m\in \mathbb{N}$.
\end{lemma}
\begin{proof}
We can write $2^m+1 = 2^{3s+r}+1 $ for $m,r,s\in\mathbb{N}\cup \{0\}$ with $ 0\leq r\leq 2 $. Therefore for any $ m $, $ 2^m+1\pmod 7 = 2^r+1\pmod 7 $ for some $r\in\{0,1,2\}$. Therefore, $ 7\nmid (2^m+1)$ for all $m\in \mathbb{N}$.  
\end{proof}
\begin{lemma}
There exists a prime such that $ p\nmid (2^m+1)$ for all $m\in \mathbb{N}$.
\end{lemma}
\begin{proof}
The proof is followed from Lemma \ref{Paper1_Lemma6.2}.
\end{proof}

\begin{theorem}\label{Paper1_Theorem6.2}
Suppose $ \mathcal{G} $ is a cyclic group of order $ n = p^s $ for an odd prime $ p $ and $ n \nmid (2^m+1)$ for all $m \in \mathbb{N}$. Then the number of LCD group codes in $\mathbb{F}_2  \mathcal{G} $ is $2^{(\frac{t+1}{2})}-1$.
\end{theorem}

\begin{proof} 
From Lemma \ref{Paper1_Lemma6.1}, if $P_i\cap M \neq \phi$, $P_i\subseteq M$ for all $i\in [t]$. It is easy to see that, each $P_{i}$ is closed under square.

Next we claim that any $ P_i $ with $ \mid P_i\mid \geq 2$, for each $ h\in P_i $ implies that $ h^{-1}\notin P_i $. Suppose there exists  $m_{1}\in \mathbb{N}\cup\{0\}$ such that $h_1=g^{i2^{m_{1}} \pmod{n}} \in P_{i}$ has the inverse in $P_{i}$. So there exists  $m_{2}\in\mathbb{N}\cup \{0\} $ with $h_{2}=g^{i2^{m_{2}} \pmod{n}} \in P_{i}$ such that $h_1 h_2 = 1$. This shows that $(2^{m_{1}}+2^{m_{2}}) \equiv 0 \pmod{n}$. Since $(2,n) = 1$, so from above we get $ n \mid (2^{m_{3}}+1)$ for some $m_3\in\mathbb{N}$, which is a  contradiction.  Hence our claim is justified.

Since $n \nmid (2^m+1)$ for all $m \in \mathbb{N}$, so from the definition of the cyclotomic cosets, each cyclotomic coset $ C_{r_i}$ with more than one element, there exists another cyclotomic coset $C_{r_j}$ with the same cardinality such that $ C_{r_i}\cup C_{r_j} $ is closed under multiplication by $ 2 $ and $ -1 $. Such union of sets, we can construct a component set of $ e $ of an LCD group code. There are $ \frac{t-1}{2} +1 $ sets which are closed under multiplication by $ 2 $ and $ -1 $. Thus we get $2^{(\frac{t+1}{2})}-1$  number of such $M$. For each $ M $, we can construct exactly one LCD  group code in $\mathbb{F}_2 \mathcal{G} $. Hence the total number of LCD group codes in $\mathbb{F}_2  \mathcal{G} $ is $2^{(\frac{t+1}{2})}-1$.
\end{proof} 

\begin{theorem}\label{Paper1_Theorem6.3}
Suppose $ \mathcal{G} $ is a cyclic group of order $n = pq $ for two odd primes $p$ and $q$ and $n \nmid (2^l + 1)$ for all $l \in \mathbb{N}$. If $ q\nmid ( 2^l + 1 )$  for all $l \in \mathbb{N}$ and $p \mid (2^m + 1 )$ for some $m \in \mathbb{N}$,  then the cyclotomic coset $C_q$ is closed under inverse but $C_p$ is not. 
\end{theorem}

\begin{proof}
Let $ x\in C_q $. Then there exists $l \in \mathbb{N}\cup\{0\}$ such that $ x = q 2^{l} \pmod {n}$. To prove that $ C_q $ is closed under inverse, just we have to show that $ -x \in{C_q} $. Since $ (n, 2) = 1 $ and $p\mid (2^m + 1 )$ for some $m \in \mathbb{N}$, so there exists $ i\in \mathbb{N} $ such that $ q 2^{l + i } \pmod {n} = - x\in C_q $. Hence $ C_q $ is closed under inverse.

Let $l_1, l_2, l_3 \in\mathbb{N}\cup\{0\}$. To prove $ C_p $ is not closed under inverse, we suppose that $ x = p 2^{l_1} \pmod {n}$ has the inverse in $ C_p $. It shows that there exists $ y = p 2^{l_2}\pmod {n} $ in $ C_p $ for some $ l_2$ such that $ x + y = 0 \pmod {n} $. This shows that $ q\mid ( 2^{l_3} + 1 ) $ for some $l_3$. This makes a contradiction. This completes the proof.
\end{proof}

\begin{theorem}\label{Paper1_Theorem6.4}
Suppose $ \mathcal{G} $ is a cyclic group of order $n = pq $ for two odd primes $p$ and $ q$ and $ n \nmid (2^l+1)$ for all $l \in \mathbb{N}$. If $ p\nmid ( 2^l + 1 )$ and $ q\nmid ( 2^l + 1 )$ for all $ l\in\mathbb{N}$, then the cyclotomic cosets $ C_p $ and $ C_q $ are not closed under inverses.
\end{theorem}
 
\begin{proof}
This proof is followed by the last part of the proof of Theorem \ref{Paper1_Theorem6.3}.
\end{proof}

\begin{theorem}\label{Paper1_Theorem6.5}
Suppose $ \mathcal{G} $ is a cyclic group of order $n = pq $ for two odd primes $p$ and $ q$ and $ n \nmid (2^l+1)$ for all $l \in \mathbb{N}$. If $ p\mid (2^l + 1 )$ for some $ l\in \mathbb{N}$ and $ q\nmid (2^m + 1 ) $ for all $ m \in \mathbb{N}$, then the number of LCD group codes in $\mathbb{F}_2  \mathcal{G} $ is $ 2^{(\frac{t}{2}+1)} -1 $.
\end{theorem}

\begin{proof}
If $P_i\cap M \neq \phi$ for some $i\in [t]$, then by Lemma \ref{Paper1_Lemma6.1}, $P_i\subseteq M$. We see that each $P_{i}$ is closed under square. From Theorem \ref{Paper1_Theorem6.3}, $ P_q $ is closed under inverse. With the similar arguments of Theorem \ref{Paper1_Theorem6.2}, we conclude that any $ P_i (\neq P_q) $ with $ \mid P_i\mid >1$, for each $ h\in P_i $ implies that $ h^{-1}\notin P_i $. 

Since $n \nmid (2^l+1)$ for all $l \in \mathbb{N}$, from the definition of the cyclotomic cosets, each cyclotomic coset $ C_{r_i}$ with more than one elements, there exists another cyclotomic coset $C_{r_j}( r_i\neq q ) $ with the same cardinality such that $ C_{r_i}\cup C_{r_j} $ is closed under multiplication by $ 2 $ and $ -1 $. So we have total  $ \frac{t-2}{2} +2 = \frac{t}{2} + 1 $ number of sets which are closed under multiplication by $ 2 $ and $ -1 $. Such union of sets, we can construct a component set of $ e $ of an LCD group code.  There are $ 2^{(\frac{t}{2} + 1)}-1 $ such $M$. For each $M$, we can construct exactly one LCD  group code in $\mathbb{F}_2 \mathcal{G} $. Hence the total number of LCD group codes in $\mathbb{F}_2  \mathcal{G} $ is $ 2^{(\frac{t}{2} + 1)}-1 $.
\end{proof} 

\begin{theorem}
Suppose $ \mathcal{G} $ is a cyclic group of order $n = pq $ for two odd primes $p$ and $ q $ and $n \nmid (2^m+1)$ for all $m \in \mathbb{N}$. If $ p\nmid (2^m + 1 )$ and $ q\nmid (2^m + 1 ) $ for all $ m\in \mathbb{N}$, then the number of LCD group codes in $\mathbb{F}_2  \mathcal{G} $ is $ 2^{(\frac{t+1}{2})}-1 $.
\end{theorem}

\begin{proof}
The proof is followed from Theorem \ref{Paper1_Theorem6.4} and Theorem \ref{Paper1_Theorem6.5}. 
\end{proof}

To prove the following theorem, we define \[ U  = \left\{ d_1 \colon n = d_1\cdot d_2~~\text{and}~~   d_1\nmid (2^l + 1)~~\text{for all}~~l \in\mathbb{N}~~ \text{and}~~ d_2\mid(2^m + 1)~~\text{for some}~~ m \in \mathbb{N} \right\}\] 
and $ \mathcal{S}' = \{ C_i : i\in U \} $.
 

\begin{theorem}
Suppose $  \mathcal{G}  $ is a  cyclic group of odd order $n$ and $ n\nmid ( 2^m + 1 ) $ for all $ m\in \mathbb{N} $. Let $ t_1 $ be the cardinality of $ S' $. Then the number of cyclic LCD group codes in $ F_2  \mathcal{G}  $ is $ 2^{(\frac{t + t_1 }{2})}-1 $.
\end{theorem}

\begin{proof}
It is clear that $P_i\cap M \neq \phi$ for some $i\in [t] $. So by Lemma \ref{Paper1_Lemma6.1}, $P_i\subseteq M$. We see that each $P_{i}$ is closed under square. Consequently each element in $ S $ is closed under multiplication by $ 2 $. Now from the similar arguments of Theorem \ref{Paper1_Theorem6.3}, each element in $ S' $ is closed under multiplication by $ -1 $ and $ 2 $. By the similar arguments of Theorem \ref{Paper1_Theorem6.4}, each element in $ S\setminus S' $ is not closed under multiplication by $ -1 $. So by the definition of the cyclotomic coset for each $ C_j\in S\setminus S' $, there exists a coset $ C_i \in S\setminus S' $ with  same cardinality of $ C_j $ such that $ C_i\cup C_j $ is closed under multiplication by $ -1 $.




Clearly the cardinality of $ S\setminus S' $ is $ t - t_1 $, so we get  total $\frac{t-t_1}{2} + t_1  = \frac{t + t_1 }{2}$ number of sets  containing elements of $ \mathcal{G}$, which are closed under square and inverse. Therefore, we get total $2^{(\frac{t+t_1}{2})}-1$ such $M$. For each $M$, we can construct exactly one LCD group code in $\mathbb{F}_2 \mathcal{G} $. Hence the total number of LCD group codes in $\mathbb{F}_2  \mathcal{G} $ is $2^{(\frac{t+t_1}{2})}-1$.
\end{proof}

\section{Dimension and minimum distance}\label{Sec5}
In this section, We consider  $ \mathcal{G}$ is any finite group. We calculate the dimension and the distance of a group code by imposing some conditions on $ M $.

\begin{theorem}\label{Paper1_Theorem5.1}
Let $\mathcal{C}=e \mathbb{F}_2\mathcal{G}$ be an LCD group code such that the component set $ M $ of $e$ forms a subgroup of $ \mathcal{G} $. Then 

\begin{enumerate}
\item $\dim( \mathcal{C} )=[ \mathcal{G} :M]$. 
\item $d(\mathcal{C}) = wt(e)$.
\end{enumerate}
\end{theorem}

\begin{proof}
Let $ [ \mathcal{G} :M] =  m $ and $ \mathcal{G} /M = \{Mg_{0}=M, Mg_{1}, \dots, Mg_{m-1} \}$, where $ g_0, g_1,\dots,g_{m-1}$ are in $ \mathcal{G} $. Suppose $f_{i}\in\mathcal{C}$ is a group codeword with component set consisting all the elements of $Mg_{i}$, for all $i \in \{ 0,1, \dots, m-1 \}$.

$(1)$ Since cosets are disjoint, $f_{0}, f_{1}, \dots, f_{m-1}$ are linearly independent over $ \mathbb{F}_2$. On the other hand, let $f\in \mathcal{C} $. Then $f = e ( a_1 + a _2 + \dots + a_l)$, for some $ a_1+a_2+\dots+a_l\in\mathbb{F}_2 \mathcal{G} $. Now  for all $j \in \{1,2, \dots, l \}$, components of $ e a_j $ are in $ M g_i $ for some $i$. So $ e a_1, e a_2, \dots, e a_l $ are in $ Span \{ f_0, f_1, \dots, f_{m-1} \}$. This shows that $ f\in  Span\{ f_0, f_1, \dots, f_{m-1} \} $. Hence $\{f_{0}, f_{1}, \dots, f_{m-1}\}$ generates $\mathcal{C}$. Therefore $\{ f_{0}, f_{1}, \dots, f_{m-1}\}$ is a basis for $\mathcal{C}$. So  $\dim(\mathcal{C}) = [ \mathcal{G} :M]$. 

$(2)$ Now $wt(f_{0}) = \cdots = wt(f_{m-1})= wt(e)$. Since any two cosets are pairwise disjoint, so \[wt(\sum_{i=0}^{m-1} f_{i}) \geq wt(f_{i}) = wt(e). \] Therefore the minimum distance of $\mathcal{C}$ is $wt(e)$.
\end{proof}
\begin{theorem}\label{Paper1_Theorem5.2}
Let $\mathcal{C}=e\mathbb{F}_2  \mathcal{G} $ be an LCD group code. Suppose $g_0\not\in M$ and $M\cup \{g_0\}$ forms a subgroup of $ \mathcal{G} $. Then

\begin{enumerate}
\item $\dim(\mathcal{C})=| \mathcal{G} |-[ \mathcal{G} :M\cup \{g_0\}]$. 
\item $d(\mathcal{C}) = 2$.
\end{enumerate}
\end{theorem}
  
\begin{proof}
Since $\mathcal{C}$ is an LCD  group code, so $\mathcal{C} ^\perp $  is an LCD  group code. Since $\mathcal{C}=e \mathbb{F}_2  \mathcal{G} $, then by Theorem \ref{Paper1_Theorem3.4}, we get $\mathcal{C} ^\perp = (1-e) \mathbb{F}_2  \mathcal{G} $. So the component set of $(1-e)$ is $M 	\cup \{g_0 \}$, which forms a subgroup of $ \mathcal{G} $. 

$(1)$ By the Theorem \ref{Paper1_Theorem5.1}, $\dim (\mathcal{C} ^\perp )= [ \mathcal{G}  : M 	\cup \{g_0\}] $. Since $\mathcal{C}$ is LCD, so by the rank-nullity theorem  \[\dim(\mathcal{C})= | \mathcal{G} |-[ \mathcal{G} \colon M 	\cup \{g_0\}].\]
    
$(2)$ Let $M' = M\cup \{g_0 \}$ and $ \mid M' \mid = k $. Then by the Theorem \ref{Paper1_Theorem5.1}, $\dim(\mathcal{C}^\perp) = [ \mathcal{G} \colon M']$. Let $ \mathcal{G}/M' = \{b_1M', b_2M', \dots, b_l M'\}$, where $b_1, b_2, \dots, b_l $ are in $\mathcal{G}$. Let $f_{i}\in\mathcal{C}^\perp$ be a group codeword corresponding to the coset $b_i M'$. 

Using Equation (\ref{paper1_equn1}), without loss of generality, we map 

 \[   f_i  \rightarrow {\underbrace{0 \cdots \cdots 0}_{k(i-1)\text{-times}}\underbrace{1 \cdots \cdots 1}_{k\text{-times}}\underbrace{0 \cdots \cdots 0}_{k(l-i)\text{-times}}}.  \]

Let \[ \mathcal{H} = \begin{bmatrix} \Phi (f_1) \\ \Phi (f_2) \\ \vdots \\ \Phi (f_l) \end{bmatrix}. \] 

Each row of $ \mathcal{H}$ contains exactly consecutive $k$ positions are $1$ and remaining $lk-k$ positions are  $0$.

Now it is easily see that $ \mathcal{H}$ is a generator matrix of $\mathcal{C} ^\perp $. Therefore $ \mathcal{H}$ is a parity check matrix of $\mathcal{C}$. So  $d(\mathcal{C}) = 2$.
\end{proof}

\section{MDS group codes}\label{Sec6}
Here, we are able to connect some interesting properties between ring theory and coding theory. The relations are true for any finite group $ \mathcal{G}$.
\begin{theorem}
Let $\mathcal{C} = e \mathbb{F}_2  \mathcal{G} $ be an LCD group code. Suppose $ g_0 \not \in M$ and $M\cup \{g_0 \}$ forms a subgroup of $ \mathcal{G} $. Then $\mathcal{C}$ is maximal in $ \mathbb{F}_2  \mathcal{G} $ if and only if $\mathcal{C}$ is MDS.
\end{theorem}

\begin{proof}
Let $\mathcal{C}$ be an maximal ideal in $\mathbb{F}_2  \mathcal{G} $. From Equation (\ref{paper1_equn1}), we know that $\mathcal{C}$ is an ideal if and only if $\Phi (\mathcal{C})$ is a linear subspace in $\mathbb{F}_2^ n $. Since  $\Phi$ is an isomorphism, so $ \mathbb{F}_2 \mathcal{G} /\mathcal{C} \cong \mathbb{F}_2 ^n /\Phi(\mathcal{C})$. Since $\mathcal{C}$ is maximal, so $\mathbb{F}_2 \mathcal{G} /\mathcal{C}$ is a field. Therefore, $\dim (\mathbb{F}_2 \mathcal{G} /\mathcal{C}) = 1$ over $\mathbb{F}_2 \mathcal{G} /\mathcal{C} $.  This shows that $\dim (\Phi(\mathcal{C})) = n - 1$. Consequently  $\dim (\mathcal{C})= n - 1 $. From Theorem \ref{Paper1_Theorem5.2}, $d(\mathcal{C})= 2 $. Now $2=d= n -( n - 1)+1 $. Hence $\mathcal{C}$ is an MDS. The converse part of the theorem is followed from the similar arguments.
\end{proof} 

\begin{theorem}
Let $ \mathcal{G} $ be an even order group and $ \mathcal{C} = e \mathbb{F}_2  \mathcal{G}  $ be an LCD group code. If $ M  $ or $ M \cup \{g_0 \} $ forms a subgroup of $ \mathcal{G} $, then $\mathcal{C}$ can not be MDS.
\end{theorem}

\begin{proof}
Let $ \mathcal{C} = e \mathbb{F}_2  \mathcal{G}  $ be an LCD MDS group code. Suppose $ M $ forms a subgroup of $ \mathcal{G} $. Suppose $\mathcal{C}$ has dimension $k$ and minimum distance $ d $, then we have $ d = n-k+1 $. Since $\mathcal{C}$ is an MDS, so $ \mathcal{C}^\perp$. Therefore, the minimum distance of $ \mathcal{C}^\perp $ is $ d^\perp = n-(n-k)+1 = k+1$.

Now, we have $ d + d^\perp = n + 2 $. Also, $ \mathcal{C}^\perp $ is of the form $ (1 - e) \mathbb{F}_2  \mathcal{G}  $, therefore \[ n + 2 = d + d^\perp \leq Support(e) + Support(1 - e) \leq 2|M| .\] Therefore,  
\begin{equation}\label{paper1_equn3}
\mid M\mid \geq \frac{ n  + 2 }{2} = \frac{ n }{2} + 1.  
\end{equation}

Since $M$ is a subgroup of $ \mathcal{G} $, so from Equation (\ref{paper1_equn3}), $ M =  \mathcal{G}  $. Now $2\mid n $, so $M$ has an element of order $2$, which contradicts the Theorem \ref{Paper1_Theorem3.5}. So $\mathcal{C}$ is not an LCD MDS group code.

If $ M\cup \{g_0 \} $ forms a subgroup of $ \mathcal{G} $, then with the similar arguments, we can conclude that $ \mathcal{C} $ is not an LCD MDS group code. This completes the proof.
\end{proof}

\section{Examples}\label{Sec7}
\begin{example}
Consider $  \mathcal{ \mathcal{G} }  = \mathbb{Z}_9 $. Here the order of $  \mathcal{G}   $ is $9$ and $ 9\mid(2^3+1)$. Now we partition $\mathbb{Z}_9 $ into cyclotomic cosets. The cyclotomic cosets are $ C_0 = \{0\} $, $ C_1 = \{   1 , 2 , 4 , 5 , 7 , 8 \}$, $ C_3 = \{3 , 6\}$. So the number of cyclotomic cosets is $ 3 $. Suppose $  \mathcal{ \mathcal{G} }  = \langle  g \rangle  $. If we choose  $ M = \{g^i : i\in C_j \} $ for any $ j\in\{ 0, 1 , 3 \}$, then $ e^2 = e = \widehat e $. Hence $ \mathcal{C} = e  \mathbb{F}_2  \mathcal{ \mathcal{G} }  $ is an LCD  group code. Also if we choose $ M = \{g^i : i\in C_j\cup C_k \} $ for any $ j , k\in\{ 0, 1 , 3 \}$ or $  M = \{g^i : i\in C_0 \cup C_1 \cup C_3 \}$, then $ e^2 = e = \widehat e $. Thus for such  $ e $, we get an LCD  group code. Hence the total number of LCD group codes in $\mathbb{F}_2 \mathbb{Z}_9 $ is $ 7 = (2^3 -1 )$.  
\end{example}

\begin{example}
Let $  \mathcal{ \mathcal{G} }  = S_3 $, $ e = (1) + (123) +(132) $. Then $ e^2 = e = \widehat e $. Here component set of $ e $ is $ M = \{(1), (123), (132) \} $, which  forms a subgroup of $  \mathcal{ \mathcal{G} }  $. Then $  \mathcal{ \mathcal{G} } /M = \{ M , (12)M \}$. Consider $ f_1 = (1) + (123) + (132) $ and $ f_2 = (12) + (23) + (13) $. Here $ \{ f_1, f_2 \} $  forms a basis of $ e  \mathbb{F}_2  \mathcal{ \mathcal{G} } $. Hence $\dim (\mathcal{C}) = 2 =[ \mathcal{ \mathcal{G} } :M]$. 
\end{example}

\begin{example}
Let $  \mathcal{ \mathcal{G} }  = S_3 $, $ e = (123) +(132) $. Then $ e^2 = e = \widehat e $. Now $ \{ (123) , (132) \}\cup \{ g_0 = 1_ \mathcal{ \mathcal{G} } \}$ forms a subgroup of $  \mathcal{ \mathcal{G} }  $. Now $ \mathcal{C} = e  \mathbb{F}_2  \mathcal{ \mathcal{G} } $ is an LCD  group code. In $ \mathcal{C} $, every nonzero element has atleast $ 2 $ nonzero coefficients. Hence $ d (\mathcal{C}) = 2 $. Since $ 2 \mid | \mathcal{ \mathcal{G} } | $, $ \mathcal{C} $ is not LCD MDS group code. 
\end{example}

\section{Conclusion}\label{Sec8}
In this paper, we have accomplished an extensive study of binary group codes. First, we have characterized binary LCD group codes. Then we have established that nontrivial cyclic LCD codes of length $2^a$ do he not exit. We also have shown the number of LCD cyclic group codes for odd order group. We have found the dimension and the minimum distance for binary group codes of any finite group. We have also built a relationship between ideals and MDS group codes. Finally, we gave some examples related to the results.

\end{document}